\documentclass[11pt,a4paper]{article}	
\usepackage{fullpage}
\usepackage{amsmath}
\usepackage{amsbsy}
\usepackage{amsfonts}
\usepackage{graphicx}
\newcommand{\be}{\begin{equation}}
\newcommand{\ee}{\end{equation}}
\newcommand{\beq}{\begin{eqnarray}}
\newcommand{\eeq}{\end{eqnarray}}

\usepackage{amsthm}
\usepackage{url}

\newtheorem{theorem}{Theorem}
\newtheorem{lemma}[theorem]{Lemma}
\newtheorem{conject}[theorem]{Conjecture}
\newtheorem{obs}[theorem]{Observation}

\begin{document}

\title{All CHSH polytopes}

\author{Stefano Pironio\\
\small Laboratoire d'Information Quantique, Universit\'{e} Libre de Bruxelles (ULB), Bruxelles, Belgium\\
\small Email: \texttt{stefano.pironio@ulb.ac.be}}

\date{February 27, 2014}

\maketitle

\begin{abstract}
The correlations that admit a local hidden-variable model are described by a family of polytopes, whose facets are the Bell inequalities. The CHSH inequality is the simplest such Bell inequality and is a facet of every Bell polytope. We investigate for which Bell polytopes the CHSH inequality is also the unique (non-trivial) facet.  We prove that the CHSH inequality is the unique facet for all bipartite polytopes where at least one party has a binary choice of dichotomic measurements, irrespective of the number of measurement settings and outcomes for the other party. Based on numerical results, we conjecture that it is also the unique facet for all bipartite polytopes involving two measurements per party where at least one measurement is dichotomic. Finally, we remark that these two situations can be the only ones for which the CHSH inequality is the unique facet, i.e., any polytope that does not correspond to one of these two cases necessarily has facets that are not of the CHSH form. As a byproduct of our approach, we derive a new family of facet inequalities.
\end{abstract}

In his seminal article demonstrating the non-locality of quantum theory, John Bell introduced an inequality that must be satisfied by any local hidden-variable model, but which can be violated according to the predictions of quantum theory \cite{bell64}. Fifty years later, such inequalities are still the main tool used to study quantum non-locality and its applications \cite{bellrmp}. The simplest and most widely used such inequality is the famous CHSH inequality \cite{chsh69}. Besides the CHSH inequality, thousands of other Bell inequalities are known\footnote{See \url{http://www.faacets.com/} for an on-line repository of Bell inequalities.}, though their properties have been barely (if at all) investigated in the vast majority of cases. One way to deal with this complex zoo and make progress in our understanding of non-locality is to attempt to identify and focus first on the simplest Bell inequalities beyond the CHSH one. This is the approach that we follow here.

Mathematically, Bell inequalities correspond to the facets of certain polytopes. Specifically, consider an experiment consisting of two separated systems $A$ and $B$ (more generally, one can also consider more than two systems). Let  $P(ab|xy)$ be the joint probability to obtain outcomes $a$ and $b$ when measurements $x$ and $y$ are performed on systems $A$ and $B$, respectively. Let us consider an experimental configuration where only a finite number of distinct measurements can be performed on each system and where each measurement can only yield a finite set of possible outcomes. The experiment is thus described by a finite set $P=\{P(ab|xy)\}$ of $t$ such probabilities, which can be seen as a point $P\in \mathbb{R}^t$. The experiment is said to admit a local model \cite{bellrmp}, if these probabilities can be factorized in the form 
\be \label{loc}
P(ab|xy) = \int\!\mathrm{d}\lambda\, q(\lambda) P_\lambda(a|x)P_\lambda(b|y)\,.
\ee 
The set of $P\in\mathbb{R}^t$ admitting a decomposition of the form (\ref{loc}) is a polytope, that is a compact and convex set in $\mathbb{R}^t$ with a finite number of extreme points, called vertices, and delimited by a finite set of hyperplanes, called facets. See \cite{bellrmp} for more details. 

This polytope depends only on the number of possible measurement settings and outcomes, but not on the way we label these measurement settings and outcomes. Let $m$ denote the number of distinct measurements that can be performed on system $A$, and for each measurement $x$, let $v_x$ denote the number of different outcomes it can yield. Define similarly $n$ as the number of distinct measurements on system $B$ and $w_y$ the number of outcomes for measurement $y$. This general experimental configuration, which we call a Bell scenario, can succinctly be described by the table $[(v_1\ldots v_m),(w_1 \ldots w_n)]$, which fully specifies the number of possible measurement settings and outcomes for $A$ and $B$. Let $\mathcal{L}[(v_1\ldots v_m),(w_1 \ldots w_n)]$ denote the polytope associated to this Bell scenario, i.e., the region specified by (\ref{loc}). When there is no need to specify the Bell scenario, we write simply $\mathcal{L}$.

The facets of the local polytope $\mathcal{L}$ are hyperplanes deliminating the local region from the non-local region, i.e., they correspond to linear inequalities $\sum_{abxy} c_{abxy} P(ab|xy)\leq c$ that are satisfied by joint probabilities $P(ab|xy)$ admitting  a local model, but which could be violated by those that do not admit such a local model -- they are thus the Bell inequalities. For any given Bell scenario, there exists in principle a systematic way to compute all the corresponding facets, but it is very inefficient. We do not know in general all the facets of a given local polytope $\mathcal{L}$.

Nevertheless, we can make some general observation regarding the facial structure of local polytopes. First, since the local polytope only depends on the number of systems, measurement settings, and outcomes, but not on the way they are labeled, we can classify them and their facets in equivalent classes under such relabelings. In the following, when we talk of \emph{a} polytope or \emph{an} inequality, we mean the entire equivalence class under relabeling of systems, measurements settings, and outcomes.  

Second, the positivity conditions $P(ab|xy)\geq 0$ are facets of every Bell polytope $\mathcal{L}$ \cite{pironio05}. Clearly, they are trivial Bell inequalities that can never be violated by any non-local point. In the following, when we refer to a Bell inequality, we thus mean a non-trivial inequality not of the positivity form $P(ab|xy)\geq 0$. We also always assume that Bell scenarios consist of at least two systems, each of which with at least two measurements, and that each measurement has at least two possible results. Indeed, if this is not the case, all facets of the corresponding polytope are either trivial positivity facets or they are equivalent to facets of polytopes corresponding to simpler Bell scenarios \cite{pironio05}. 

The smallest non-trivial local polytope is thus $\mathcal{L}[(2\, 2),(2\, 2)]$. It has a unique (non-trivial) facet, the well-known CHSH inequality \cite{chsh69}
\be 
\sum_{x,y=0}^{1}\sum_{a,b=0}^{1} (-1)^{a+b+xy}P(ab|xy)\leq 2\,.
\ee
The CHSH inequality is also a facet of every (non-trivial) local polytope. Indeed, if an inequality is a facet of a given polytope $L$, it can always be lifted to a polytope $\mathcal{L}'$ with more systems, measurements, or outcomes in such a way that it defines a facet of $\mathcal{L}'$ \cite{pironio05}. For instance, to lift the CHSH inequality to a situation with 3 possible outcomes per measurement, one simply groups the outcomes $1$ and $2$ in an effective ``$1$" outcome. This yields a facet of $\mathcal{L}[(3\, 3),(3\, 3)]$, which is essentially equivalent to the original $2$-outcome CHSH inequality. Using the same procedure, one can lift CHSH to an arbitrary number of outcomes. Similarly, one can introduce ways to lift the CHSH inequality to scenarios with more measurements (in this case simply by ignoring the additional measurements) or more systems.
	
If the CHSH inequality is a facet of every local polytope, then the following question arises: for which local polytopes is it also the unique facet? This is the question that we attempt to solve here. An answer to this question implies that we also know what are the simplest Bell scenarios for which the CHSH inequality is not the unique facet and it thus allow us identifying the simplest Bell inequalities beyond CHSH.

The main result of this article is the following one, which is proven further below. 
\begin{theorem}\label{theo}
All non-trivial facet inequalities of $\mathcal{L}[(2\,2),(w_1\,w_2\ldots w_n)]$ for any $n\geq 2$ and any $w_i\geq 2$ are CHSH inequalities.
\end{theorem}
For the special case $w_i=2$, this result was already obtained in \cite{sli03,CG}. It shows that as long as one party has only a choice between two dichotomic measurements then they are no other Bell inequalities than CHSH, irrespective of the number of measurements and outcomes for the other party.

We also propose the following conjecture.
\begin{conject}
All non-trivial facet inequalities of $\mathcal{L}[(2\,v_{2}),(w_{1}\,w_{2})]$ for any $v_2,w_1,w_2\geq 2$ are CHSH inequalities.
\end{conject} 
Note first that this conjecture is true for the particular case where $v_2=2$, since it then follows from Theorem~\ref{theo}. To provide evidence for this conjecture for more general cases, we have checked using polytope software \cite{porta}, that it is indeed satisfied for all $2\leq v_2,w_1,w_2 \leq 5$. Another (more intuitive) argument for this conjecture is that the unique non-local vertices of the no-signalling polytope associated with the Bell scenario $[(2\,v_{2}),(w_{1}\,w_{2})]$ are the the PR-boxes, which maximally violate the CHSH inequality \cite{boxes}.

Finally, we make the following observation.
\begin{obs}
All non-trivial Bell scenarios that are not of the form $[(2\,2),(w_1\,w_2\ldots w_n)]$ or $[(2\,v_{2}),(w_{1}\,w_{2})]$ have facet inequalities that are inequivalent to the CHSH inequality.
\end{obs}
By the lifting property of facet inequalities, it is sufficient to verify this statement for the simplest Bell scenarios not of the above two forms. There are four possible cases.
\begin{itemize}
\item $[(2\, 2),(2\,2),(2\,2)]$. Since the above Bell scenarios involve two systems, one first possibility is to consider a tripartite Bell scenario. The simplest one is $[(2\, 2),(2\,2),(2\,2)]$, which is known to contain $44$ non-trivial facets in addition to the CHSH inequality \cite{sli03}.
\item $[(2\,2\,2),(2\,2\,2)]$. Going back to bipartite Bell scenarios, the next possibility is to consider a scenario with three measurement choices for each system, since the above two Bell scenarios have always at least a binary choice of measurements for one system. The simplest such Bell scenario is $[(2\,2\,2),(2\,2\,2)]$, which contains the Froissart inequality in addition to the CHSH inequality \cite{froissard81}.
\item $[(3\,3),(3\,3)]$. If we restrict to binary choices of measurements on each system, then they must all have at least $3$ outcomes. The simplest such scenario is $[(3\,3),(3\,3)]$, which contains the CGLMP inequality \cite{cglmp}.
\item $[(2\,3),(2\,2\,2)]$. Finally, the last possibility is to have a binary choice of measurements in only one system, in which case one of these measurements must have at least $3$ outcomes. The simplest such Bell scenario is $[(2\,3),(2\,2\,2)]$, whose facial structure had not been determined before. We have determined all the facets of this polytope using the polytope software \texttt{cdd} \cite{cdd}. It contains the following facet inequality which is not of the CHSH form
\be\label{newineq}
\begin{split}
P_A(1|1)+P_B(1|1)+P_B(1|2)&-P(11|11)-P(11|12)-P(11|21)\\
&-P(21|22)
-P(11|13)+P(11|23)+P(21|23)\geq 0\,,
\end{split}
\ee
\end{itemize}
where the possible values of $a,b,x,y$ are labeled as $1,2,3$ and we have written $P_A(a|x)$ for the marginal probabilities on system $A$ and similarly $P_B(b|y)$ for the marginal probabilities on system $B$ (these marginals are well-defined thanks to the no-signalling principle \cite{bellrmp}).

In summary, if Conjecture~2 is true, the simplest Bell inequalities beyond CHSH are the CGLMP inequality, Froissard's inequality, the inequalities associated with the tripartite situation, and the new inequality (\ref{newineq}). Since any lifting of any one of these facet inequalities to a more complex Bell scenario is also a facet of the corresponding polytope, all Bell scenarios that are not of the form $[(2\,2),(w_1\,w_2\ldots w_n)]$ or $[(2\,v_{2}),(w_{1}\,w_{2})]$ must contain at least one of these four types of facets. In principle, one could then also ask the question that we asked for the CHSH inequality: for which Bell scenarios are such facets the only facets?

We note that based on an early draft of the present paper, some properties of the new inequality (\ref{newineq}) have been already investigated in \cite{dimwit}. In particular, the maximal quantum violation of the above inequality is $-0.2532$, which is obtained by measuring a partially entangled state of two qutrits. The maximal quantum violation with qubits, however, is $-0.2071$, i.e., the inequality (\ref{newineq}) can serve as a dimension witness for qutrits.

Finally, we conclude by generalizing this inequality to the class of local polytopes $\mathcal{L}_n=\mathcal{L}[(2\,n),(2\ldots 2)]$ with $n$ binary measurements on system $B$ in the following way
\be\label{ineq2}
\begin{split}
P_A(1|1)+\sum_{k=1}^{n-1}P_B(1|k)&-\sum_{k=1}^{n}P(11|1k)-\sum_{k=1}^{n-1}P(k1|2k)+\sum_{k=1}^{n-1}P(k1|2n)\geq 0\,,
\end{split}
\ee
where the possible values of $a,b,x,y$ are labeled as $1,2,\ldots,n$. For $n=2$, we recover the CHSH inequality, and for $n=3$ the inequality (\ref{newineq}). More generally, we have the following result, whose proof is presented below the one of Theorem~1.
\begin{theorem}
The inequality (\ref{ineq2}) is a facet of $\mathcal{L}_n$ for any $n\geq 2$. 
\end{theorem}
We leave open the question of determining the properties of the inequalities (\ref{ineq2}), such as their quantum violations.

\section*{Proof of Theorem 1}
A standard procedure to solve facet enumeration problems is the Fourier-Motzkin elimination method \cite{Schrijver89}. It is this approach that we will use to construct all the inequalities of $\mathcal{L}[(2\,2),\linebreak[4](w_1\,w_2\ldots w_n)]$. Let us first remind that a set of probabilities $P=\{P(ab|xy)\}$ admits a local model (\ref{loc}) if and only if there exists a joint distribution $P(a_1\ldots a_m b_1\ldots b_n)$ that is positive $P(a_1\ldots a_m b_1\ldots b_n)\geq 0$ and normalized $\sum_{a_1,\ldots, a_m, b_1,\ldots,b_n}P(a_1\ldots a_m b_1\ldots b_n)=1$ and which  returns the probabilities $P(ab|xy)$ as marginals \cite{bellrmp}:
\be \label{locjoint}
P(ab|xy)=\sum_{a_1}\cdots\sum_{a_m}\sum_{b_1}\cdots\sum_{b_n} P(a_1\ldots a_mb_1\ldots b_n)\, \delta_{a_x,a}\,\delta_{b_y,b}\,.
\ee
In quantum mechanics, such joint distributions are ill-defined for incompatible (non-commuting) measurements, hence the origin of the contradiction with local models. Fine further noticed that in the case of $\mathcal{L}[(2\,2),(2\,2)]$ the existence of a complete joint distributions for only \emph{one} of the two systems is already equivalent to the existence of a local model \cite{Fine82}. The following lemma extends Fine's results to more measurements and outcomes.
\begin{lemma}\label{fine} Their exists a joint distribution $P(a_1\ldots a_m b_1\ldots b_n)$  satisfying (\ref{locjoint}) if and only if there exists $n$ probability distributions $P(a_1\ldots a_mb|y)$, one for each $y=1,\dots,n$, with the following two properties:
\begin{enumerate}\renewcommand{\labelenumi}{\roman{enumi})}
\item they return
the original correlations as marginals:
\begin{equation}\label{theojointeq1}
P(ab|xy)=\sum_{a_1}\cdots\sum_{a_m}P(a_1\ldots a_m b|y)\, \delta_{a_x,a}\,.
\end{equation}
\item they yield one and the same joint distribution $P(a_1\ldots a_m)$ on Alice's side:
\begin{equation}\label{theojointeq2}
\sum_{b}P(a_1\ldots a_m b|y)=P(a_1\ldots a_m).
\end{equation}
\end{enumerate}
\end{lemma}
\begin{proof}
The necessary condition is straightforward, just observe that the $n$ distributions of the lemma can be obtained from $P(a_1\ldots a_m b_1\ldots b_n)$ as marginals:
\begin{equation}
P(a_1\ldots a_m b|y)=\sum_{y'}\sum_{b_{y'}}P(a_1\ldots a_m b_1\ldots b_n)\delta_{b,b_y}\,.
\end{equation}
To show sufficiency, set
\begin{equation}
P(a_1\ldots a_m b_1\ldots b_n)=\left\{\begin{array}{cl}\displaystyle\frac{\prod\limits_yP(a_1\ldots a_m b_y|y)}{{\left(P(a_1\ldots a_m)\right)}^{m-1}}&\mbox{if }P(a_1\ldots a_m)\neq 0\vspace{0.5em}\\0&\mbox{if }P(a_1\ldots a_m)=0.\end{array}\right.
\end{equation}
It is straightforward to verify, using \eqref{theojointeq1} and \eqref{theojointeq2}, that the resulting joint distribution is positive and normalized and that it satisfies (\ref{locjoint}).
\end{proof}
Instead of considering the original condition (\ref{locjoint}) which involves the $\left(\prod_{x=1}^{m}v_x\right)\left(\prod_{y=1}^{n}w_y\right)$ probabilities $P(a_1\ldots a_m b_1\ldots b_n)$, we can thus instead consider the conditions (\ref{theojointeq1}) and (\ref{theojointeq2}) which involve only $\left(\prod_{x=1}^{m}v_x\right)\left(\sum_{y=1}^{n}w_y\right)$ probabilities $P(a_1\ldots a_m b|y)$, an exponential decrease in the number of unknowns. This simplification will allow us to prove Theorem~1.

Let us now focus specifically on the case  $a,x\in\{1,2\}$, $b\in\{1,\ldots,w_y\}$, $y\in\{1,\ldots,n\}$. By the above lemma, the set of $P(ab|xy)$ that admit a local model is the set defined by the following linear system
\begin{align}
P(ab|1y)=&\sum_{\phantom{a_1,}a_2\phantom{,b}} P(aa_2b|y)\label{eq1}\\
P(ab|2y)=&\sum_{\phantom{a_2,}a_1\phantom{,b}} P(a_1ab|y)\label{eq2}\\
 &\phantom{\sum_{a_1,a_2,b}}P(a_1a_2b|y)\geq 0\\
 &\sum_{a_1,a_2,b}P(a_1a_2b|y)=1\label{eq3}\\
 &\sum_{\phantom{a_1,}b\phantom{,a_2}}P(a_1a_2b|y)=P(a_1a_2)\label{coupling}\,,
\end{align}
with unknowns $P(a_1a_2b|y)$
This is thus an ensemble of $n$ linear inequality systems, one for each $y$, coupled through Eq. (\ref{coupling}).
The Fourier-Motzkin elimination method consists of successively eliminating each unknown $P(a_1a_2b|y)$, until one is left only with linear conditions on the $P(ab|xy)$, which are thus the Bell inequalities. It is analogous to the way systems of linear equalities are solved by the Gaussian method. Contrary to Gaussian elimination, Fourier-Motzkin elimination, however, is not efficient in general. In the above specific case, it can nevertheless be carried out to obtain the complete set of Bell inequalities.

The first step in the Fourier-Motzkin elimination is to solve the equality constraints in the above equations, i.e., express a subset of the variables as linear combination of the remaining variables. Formally, the above set of equality constraints is equivalent to the set of equalities defining the no-signalling tripartite polytope, where we identify $P(a_1a_2b|y)$ with an usual distribution $P(a_1a_2b|12y)$ of a tripartite Bell scenario. It is well-known that a minimal representation of the no-signalling polytope that does not involve any linear constraints between the remaining variables is obtained by considering all probabilities (including the one-partite and two-partite marginals) for which one of the outcomes, say the last one, does not appear. It follows that we can choose as a set of independent variables $\{P_A(1|1)$, $P_A(1|2)$, $P_B(b|y)$, $P(1b|1y), P(1b|2y)$, $P(11)$, $P(11b|y)\}$ where $b<w_y$ and where $P_A(1|1)=\sum_{b=1}^{w_y}P(1b|1y)$, $P_A(1|2)=\sum_{b=1}^{w_y}P(1b|2y)$, $P_B(b|y)=\sum_{a=1}^{2}P(ab|xy)$. In terms of these variables we can express the remaining variables, i.e. the $P(a_1a_2b|y)$ for which at least one of the $a_i=2$ or $b=w_y$, as
\beq
P(11w_y|y)&=&P(11)-\sum_{b=1}^{w_y-1}P(11b|y)\label{dep1}\\
P(12b|y)&=&P(1b|1y)-P(11b|y)\\
P(21b|y)&=&P(1b|2y)-P(11b|y)\\
P(12w_y|y)&=&P_A(1|1)-P(11)-\sum_{b=1}^{w_y-1}P(1b|1y)+\sum_{b=1}^{w_y-1}P(11b|y)\\
P(21w_y|y)&=&P_A(1|2)-P(11)-\sum_{b=1}^{w_y-1}P(1b|2y)+\sum_{b=1}^{w_y-1}P(11b|y)\\
P(22b|y)&=&P_B(b|y)-P(1b|1y)-P(1b|2y)+P(11b|y)\\
P(22w_y|y)&=&1-P_A(1|1)-P_A(1|2)-\sum_{b=1}^{w_y-1}P_B(b|y)\\
&&+P(11)+\sum_{b=1}^{w_y-1}P(1b|1y)+\sum_{b=1}^{w_y-1}P(1b|2y)-\sum_{b=1}^{w_y-1}P(11b|y)\label{depf}\,.
\eeq
One can also simply verify directly that the above set of equality constraints are consistent with (\ref{eq1}), (\ref{eq2}), (\ref{eq3}), (\ref{coupling}) and that these later equalities have all been taken into account, i.e., there is no linear dependency left on the variables appearing on the right-hand side of Eqs. (\ref{dep1})-(\ref{depf}).

Now that we have taken into account all equality constraints, the only constraints that remain are the inequalities $P(a_1a_2b|y)\geq 0$. In terms of our set of independent variables, they simply correspond to $P(11b|y)\geq 0$ for $b<w_y$, together with the condition that all the right-hand side of Eqs. (\ref{dep1})-(\ref{depf}) should be positive. To write down explicitly these inequalities, let us first introduce the following  change of variable that will simplify the subsequent analysis
\be
S(b|y)=\sum_{b'=1}^{b} P(11b'|y)\,.
\ee
This implies $P(111|y)=S(1|y)$ and $P(11b|y)=S(b|y)-S(b{-}1|y)$ for $b>1$. The inequalities $P(a_1a_2b|y)\geq 0$ can then be written as
\begin{align}
S(1|y)&\geq 0\label{or1}\\
S(1|y)&\geq -P_B(1|y)+P(11|1y)+P(11|2y)\\
S(1|y)&\leq P(11|1y)\\
S(1|y)&\leq P(11|2y)\label{or2}\,,
\end{align}
\begin{align}
S(b{-}1|y)\leq& S(b|y)\label{or3}&(b=2,\ldots,w_y-1)\\
S(b{-}1|y)\leq& S(b|y)+P_B(b|y)-P(1b|1y)-P(1b|2y)&(b=2,\ldots,w_y-1)\\
S(b{-}1|y)\geq&S(b|y)-P(1b|1y)&(b=2,\ldots,w_y-1)\\
S(b{-}1|y)\geq&S(b|y)-P(1b|2y)\label{or4}&(b=2,\ldots,w_y-1)\,,
\end{align}
and
\begin{align}
S(w_y-1|y)&\leq P(11)\label{or5}\\
S(w_y-1|y)&\leq  1-P_A(1|1)-P_A(1|2)-\sum_{b=1}^{w_y-1}P_B(b|y)+P(11)\nonumber\\
&+\sum_{b=1}^{w_y-1}P(1b|1y)+\sum_{b=1}^{w_y-1}P(1b|2y)\\
S(w_y-1|y)&\geq-P_A(1|1)+P(11)+\sum_{b=1}^{w_y-1}P(1b|1y)\\
S(w_y-1|y)&\geq -P_A(1|2)+P(11)+\sum_{b=1}^{w_y-1}P(1b|2y)\label{or6}\,.
\end{align}
The first four inequalities correspond to the conditions $P(111|y)\geq 0$, $P(221|y)\geq 0$, $P(121|y)\geq 0$, $P(211|y)\geq 0$,  the four following series of inequalities to $P(11b|y)\geq 0$, $P(22b|y)\geq 0$, $P(12b|y)\geq 0$, $P(21b|y)\geq 0$ with $b=2,\ldots,w_y-1$, and finally the last four inequalities to  $P(11w_y|y)\geq 0$, $P(22w_y|y)\geq 0$, $P(12w_y|y)\geq 0$, $P(21w_y|y)\geq 0$.

We can now use use the Fourier-Motzkin method on the above system of inequalities to eliminate successively all the $S(b|y)$. The process consists in combining, for each $S(b|y)$, the inequalities of the form $S(b|y)\leq C$ with those of the form $S(b|y)\geq D$ to obtain the condition $C\geq D$ independent of $S(b|y)$. To do this, let us introduce a further notation. Let $G_b$ be a subset of $\{1,\ldots,b\}$, possibly empty, and $\overline{G_b}=\{1,\ldots,b\}\setminus G_b$ be the complementary set and write
\be 
P(G_b|y)=\sum_{b'\in G_b}P(b'|y)\,,\qquad P(1G_b|xy)=\sum_{b'\in G_b}P(1b'|xy)\,.
\ee
We now apply the Fourier-Motzkin elimination. Suppose first that  at some step in the Fourier-Motzkin iteration, all the inequalites involving $S(b|y)$ are of the following form
\beq 
\label{s1}S(b|y)&\geq&-P(G_b|y)+P(1G_b|1y)+P(1G_b|2y)\qquad \forall\, G_b\subseteq \{1,\ldots,b\}\\
\label{s2}S(b|y)&\geq& S(b+1|y)-P(1(b+1)|1y)\\
\label{s3}S(b|y)&\geq& S(b+1|y)-P(1(b+1)|2y)\\
\label{s4}S(b|y)&\leq& P(1G_b|1y)+P(1\bar G_b|2y) \qquad\forall\, G_b\subseteq \{1,\ldots,b\}\\
\label{s5}S(b|y)&\leq& S(b+1|y)\\
\label{s6}S(b|y)&\leq& S(b+1|y)+P(b+1|y)-P(1(b+1)|1y)-P(1(b+1)|2y)\,.
\eeq
We show here below that if we eliminate $S(b|y)$, we remain with a set of inequalities that are of the above form but with $b$ replaced by $b+1$. Since initially all inequalities involving $S(1|y)$ are of the above form, by sequentially removing the variables $S(1|y)$, $S(2|y),\ldots,$ one thus get a closed iterative process that finishes at the $(w_y-2$)th step, where we are left with inequalities involving only $S(w_y-1|y)$.

Let us now verify the inductive property mentioned above. Eliminating $S(b|y)$ in (\ref{s1})-(\ref{s6}) by combining $S(b|y)\geq X$ and $S(b|y)\leq Y$ to get $X\leq Y$, we are left with the following inequalities 
\beq
\label{sp1}S(b+1|y)&\geq& -P(G_{b+1}|y)+P(1G_{b+1}|1y)+P(1G_{b+1}|2y)\quad \forall\, G_{b+1}\subseteq \{1,\ldots,b+1\}\\
\label{sp2}S(b+1|y)&\leq& P(1G_{b+1}|1y)+P(1\bar G_{b+1}|2y) \quad \forall\, G_{b+1}\subseteq \{1,\ldots,b+1\}
\eeq
Inequality (\ref{sp1}) follows by combining (\ref{s1}) with (\ref{s5}) and with (\ref{s6}).
Inequality (\ref{sp2}) follows by combining (\ref{s4}) with (\ref{s2}) and with (\ref{s3}).
Other combination such as combining (\ref{s1}) and (\ref{s4}) yields inequalities of the form
\be
P(G_b|y)-P(1(G_b\setminus G_{b'})|1y)-P(1(G_b\cap G_{b'})|2y)+P(1(G_{b'}\setminus G_b)|1y)+P(1(\bar G_{b'}\setminus G_b)|2y)\geq 0\,.
\ee
Such inequalities only involve the probabilities $P(ab|xy)$. It can be checked that they are redundant with the positivity inequalities $P(ab|xy)\geq 0$ and we thus do not need to keep track of them.

The inequalities (\ref{sp1}) and (\ref{sp2}), on the other hand, involve the variables $S(b+1|y)$. These inequalities, together with the original inequalities (\ref{or1})-(\ref{or6}) for $S(b+1|y)$ are of the form (\ref{s1})-(\ref{s6}). This establishes the inductive nature of our iterative process. 
After the $(w_y-2)$th step, we are left with the inequalities (\ref{or5})-(\ref{or6}) and the inequalities (\ref{sp1})-(\ref{sp2}) for $S(w_y-1|y)$. Eliminating this last variable yields
\beq
P(11)&\geq & -P_B(G_{w_y-1}|y)+P(1G_{w_y-1}|1y)+P(1G_{w_y-1}|2y)\\
P(11)&\geq & -1+P_A(1|1)+P_A(1|2)+P(G_{w_y-1}|y)-P(1G_{w_y-1}|1y)-P(1G_{w_y-1}|2y)\\
P(11)&\leq & P_A(1|1)-P(1G_{w_y-1}|1y)+P(1G_{w_y-1}|2y)\\
P(11)&\leq &P_A(1|2)+P(1G_{w_y-1}|1y)-P(1G_{w_y-1}|2y)\,,
\eeq
for any possible subsets $G_{w_y-1}$ of $\{1,\ldots,w_y-1\}$.
It now remains to eliminate the variable $P(11)$. Note that if there is only one choice of measurement on the second system, all Bell inequalities are trivial inequalities arising from the positivity constraints. This means that combining two above inequalities that involve the same value $y$ gives rise to trivial inequalities. This fact can also be directly verified.  The only non-trivial inequalities are thus obtained when combining two of the above inequalities  with two different values $y$ and $y'$. There are thus four types of possibilities that lead to 
\beq
1\geq P_A(1|1)+P_B(G_y|y)-P(1G_y|1y)-P(1G_y|2y)-P(1G'_{y'}|1y')+P(1G'_{y'}|2y')\geq 0\\
1\geq P_A(1|2)+P_B(G_y|y)-P(1G_y|1y)-P(1G_y|2y)+P(1G'_{y'}|1y')-P(1G'_{y'}|2y')\geq 0
\eeq
for any $y$, $y'$ and any non-empty subsets $G_y\subseteq\{1,\ldots,w_y-1\}$ and $G'_{y'}\subseteq\{1,\ldots,w_{y'}-1\}$. These inequalities correspond to all possible liftings of the two-input two-output CHSH inequalities
\be
1\geq P_A(1|1)+P_B(1|1)-P(11|11)-P(11|21)-P(11|12)+P(11|22)\geq 0
\ee
written in the CH form \cite{CH}. Such lifting are obtained by associating in the above inequality, the measurement choices $y$ and $y'$ with the choices $1$ or $2$ and by grouping subset $G_y$ and $G'_{y'}$ of the corresponding outcomes in an effective ``$1$" outcome, and all the remaining outcomes in an effective ``$2$" outcome. This achieves the proof of Theorem 1.

\section*{Proof of Theorem 4}
We remind \cite{bellrmp} that to show that (\ref{ineq2}) is a facet of the local polytope $\mathcal{L}_n$, it is sufficient to show that $i)$ it is satisfied by all local deterministic points (i.e. all vertices of $\mathcal{L}_n$) 
\be 
P(ab|xy)=\delta_{a\alpha_x}\delta_{b\beta_y}\,,
\ee
and $ii)$ that there are $\dim \mathcal{L}_n$ affinely deterministic points that saturate it. 

Let us start by showing $i)$. There are two cases to consider depending on the value of $\alpha_{1}$. If $\alpha_1=1$, then the Bell expression on the left-hand side of (\ref{ineq2}) is equal to $1-\sum_{k=1}^{n} P(k1|2k)$, which is clearly positive for any $\alpha_2=k'$. If $\alpha_1=2$, the Bell expression is equal to $\sum_{k=1}^{n-1}P_B(1|k)-\sum_{k=1}^{n-1} P(k1|2k)+\sum_{k=1}^{n-1}P(k1|2n)$, which is also clearly positive since $P_B(1|k)\geq P(k1|2k)$ for any $k$.

Let us now show $ii)$. First note that the dimension of $\mathcal{L}_n$ is $n(n+2)$ \cite{pironio05}. Consider now the following $n(n+2)$ local deterministic points grouped in four subsets:
\begin{align}
(1)\quad&\alpha_{1}=1,\,\alpha_{2}=k,\,\beta_{l}=1\text{ for all }l &(k=1,\ldots,n)\nonumber\\
(2)\quad&\alpha_{1}=1\,\alpha_{2}=k,\,\beta_{l}=2,\,\beta_{j}=1 \text{ for all }j\neq l&(k,l=1,\ldots,n
\text{ and }k\neq l)\nonumber\\
(3)\quad&\alpha_{1}=2,\,\alpha_{2}=k,\,\beta_{l}=2\text{ for all }l &(k=1,\ldots,n)\nonumber\\
(4)\quad&\alpha_{1}=2,\,\alpha_{2}=k,\,\beta_{k}=1,\,\beta_{l}=2 \text{ for all }l\neq k&(k=1,\ldots,n)\nonumber
\end{align}
It is easily verified that they all saturate the inequality (\ref{ineq2}). Furthermore, they are all affinely independent. Indeed, first note that the vertices from the first set are affinely independent. If the vertices of the second, third and fourth set are then successively added, we obtain a resulting set where all points are affinely independent because each newly introduced vertex has a nonzero component which is equal to zero for all the previously introduced vertices. For the points in the second set, this component is $P(k2|2l)$, for the third, $P(k2|2k)$, and for the fourth, it is $P(21|1k)$.

\paragraph{Acknowledgments.}
We acknowledge financial support from the European Union under the projects DIQIP and QALGO, from 
the F.R.S.-FNRS under the project DIQIP, and from the Brussels-Capital Region through a BB2B grant.


\end{document}